\documentclass[conference]{IEEEtran}
\IEEEoverridecommandlockouts
\usepackage{cite}
\usepackage{amsmath,amssymb,amsfonts,amsthm}
\usepackage{algorithmic}
\usepackage{graphicx}
\usepackage{textcomp}
\usepackage{float}
\usepackage{caption}
\def\BibTeX{{\rm B\kern-.05em{\sc i\kern-.025em b}\kern-.08em
    T\kern-.1667em\lower.7ex\hbox{E}\kern-.125emX}}
\theoremstyle{plain}
\newtheorem{theorem}{Theorem}[section]

\newtheorem{assumption}{Assumption}[section]

\theoremstyle{definition}

\theoremstyle{remark}
\newtheorem{remark}{Remark}

\begin{document}

\title{Cooperative control of multi-agent systems to locate source of an odor \\
\thanks{A. Sinha is with School of Mechatronics \& Robotics, Indian Institute of Engineering Science and Technology; and Central Scientific Instruments Organization (CSIR- CSIO), India.\newline\textbf{email: }\texttt{sinha.abhinav.pg2016@mechatronics.iiests.ac.in}\newline R. Kaur, R. Kumar \& A. P. Bhondekar are with CSIR- CSIO. \newline\textbf{emails: }\texttt{rishemjit.kaur@csio.res.in,riteshkr@csio.res.in,\newline amolbhondekar@csio.res.in}}
}

\author{Abhinav Sinha, Rishemjit Kaur, Ritesh Kumar and Amol P. Bhondekar}

\maketitle

\begin{abstract}
This work targets the problem of odor source localization by multi-agent systems. A hierarchical cooperative control has been put forward to solve the problem of locating source of an odor by driving the agents in consensus when at least one agent obtains information about location of the source. Synthesis of the proposed controller has been carried out in a hierarchical manner of group decision making, path planning and control. Decision making utilizes information of the agents using conventional Particle Swarm Algorithm and information of the movement of filaments to predict the location of the odor source. The predicted source location in the decision level is then utilized to map a trajectory and pass that information to the control level. The distributed control layer uses sliding mode controllers known for their inherent robustness and the ability to reject matched disturbances completely. Two cases of movement of agents towards the source, i.e., under consensus and formation have been discussed herein. Finally, numerical simulations demonstrate the efficacy of the proposed hierarchical distributed control.
\end{abstract}


\section{Introduction}
\subsection{Overview}
Inspiration of odor source localization problem stems from behavior of biological entities such as mate seeking by moths, foraging by lobsters, prey tracking by mosquitoes and blue crabs, etc., and is aimed at locating the source of a volatile chemical. These behaviors have long been mimicked by autonomous robot(s). Chemical source tracking has attracted researchers around the globe due to its applications in both civilian and military domains. A plethora of applications are possible, some of which include detection of forest fire, oil spills, release of toxic gases in tunnels and mines, gas leaks in industrial setup, search and rescue of victims and clearing leftover mine after an armed conflict.
A plume containing filaments, or odor molecules, is generally referred to the downwind trail formed as a consequence of mixing of contaminant molecules in any kind of movement of air. The dynamical optimization problem of odor source localization can be effectively solved using multiple robots working in cooperation. The obvious advantages of leveraging multi-agent systems (MAS) are increased probability of success, redundancy and improved overall operational efficiency and spatial diversity in having distributed sensing and actuation. 
\subsection{Motivation}
Odor source localization is a three stage problem-- sensing, maneuvering and control. Some of reported literature on odor source localization date back to 1980s when Larcombe et al. \cite{Larcombe1984} discussed such applications in nuclear industry by considering a chemical gradient based approach. Other works in 1990s \cite{240354,594225,Buscemi,414920,Russell:2000:ODM:518667} relied heavily on sensing part using techniques such as chemotaxis \cite{RUSSELL200383}, infotaxis \cite{vergassola:hal-00326807}, anemotaxis \cite{1245262,1703649} and fluxotaxis \cite{Zarzhitsky:2008:PAC:1571283}. The efficiency of such algorithms was limited by the quality of sensors and the manner in which they were used. These techniques also failed to consider turbulence dominated flow and resulted in poor tracking performance.\medskip

\noindent Bio-inspired algorithms have been reported to maneuver the agents, some of which include Braitenberg style \cite{braitenberg}, E. coli algorithm \cite{Lytridis:2001:ONS:770673.770677}, Zigzag dung beetle approach \cite{ISHIDA1994153}, silkworm moth style \cite{russell_silkworm,862824,MARQUES200251} and their variants. A tremendous growth of research attention towards cooperative control has been witnessed in the past decade \cite{1431045,5711689} but very few have addressed the problem of locating source of an odor. Hayes et al. \cite{Hayes:2003:SRO:976910.976917} proposed a distributed cooperative algorithm based on swarm intelligence for odor source localization and experimental results proved multiple robots perform more efficiently than a single autonomous robot. A Particle Swarm Optimization (PSO) algorithm \cite{488968} was proposed by Marques et al. \cite{Marques2006,4168420} to tackle odor source localization problems. To avoid trapping into local maximum concentrations, Jatmiko et al. \cite{4168420} proposed modified PSO algorithms based on electrical charge theory, where neutral and charged robots has been used. Lu et al. \cite{5586502} proposed a distributed coordination control protocol based on PSO to address the problem. It should be noted that simplified PSO controllers are a type of proportional-only controller and the operating region gets limited between global and local best. This needs complicated obstacle avoidance algorithms and results in high energy expenditure. Lu et al. \cite{6119291} also proposed a cooperative control scheme to coordinate multiple robots to locate odor source in which a particle filter has been used to estimate the location of odor source based on wind information, a movement trajectory has been planned, and finally a cooperative control scheme has been proposed to coordinate movement of robots towards the source. \medskip

\noindent Motivated by these studies, we have implemented a robust and powerful hierarchical cooperative control strategy to tackle the problem. First layer is the group level in which the information about the source via instantaneous sensing and swarm intelligence is obtained. Second layer is designed to maneuver the agents via a simplified silkworm moth algorithm. Third layer is based on cooperative sliding mode control and the information obtained in the first layer is passed to the third layer as a reference to the tracking controller.
\subsection{Contributions}
Major contributions of this paper are summarized below.
\begin{enumerate}
\item As opposed to existing works on cooperative control to locate source of odor, we have considered a more general formulation by taking nonlinear dynamics of MAS into account. When the uncertain function is zero, the problem reduces to stabilizing integrator dynamics.\\
\item The control layer is designed on the paradigms of sliding mode, a robust and powerful control with inherent robustness and disturbance rejection capabilities. The reaching law, as well as the sliding manifold in this study are nonlinear and novel resulting in smoother control and faster reachability to the manifold. Use of sliding mode controller also helps in achieving a finite time convergence as opposed to asymptotic convergence to the equilibrium point. The proposed control provides stability and ensures robustness even in the presence of bounded disturbances and matched uncertainties.\\
\item Odor propagation is non-trivial, i.e., odor arrives in packets, leading to wide fluctuations in measured concentrations. Plumes are also dynamic and turbulent. As odor tends to travel downwind, direction of the wind provides an effective information on relative position of the source. Hence, we have used wind information based on a measurement model describing movement of filaments and concentration information from swarm intelligence to locate the source of odor.\\
\item Formation keeping of agents to locate source of odor has also been demonstrated in this work.
\end{enumerate}
\subsection{Paper Organization}
After introduction to the study in section I, remainder of this work in organized as follows. Section II provides insights into preliminaries of spectral graph theory and sliding mode control. Section III presents dynamics of MAS and mathematical problem formulation, followed by hierarchical distributed cooperative control scheme in section IV. Results and discussions have been carried out in section V, followed by concluding remarks in section VI.

\section{Preliminaries}

\subsection{Spectral Graph Theory for Multi-Agent Systems}
A directed graph, also known as digraph is represented throughout in this paper by $\mathcal{G} = (\mathcal{V,E,A})$. $\mathcal{V}$ is the nonempty set in which finite number of vertices or nodes are contained such that $\mathcal{V} = \{1, 2, ..., N\}$. $\mathcal{E}$ denotes directed edge and is represented as $\mathcal{E} = \{(i,j)\hspace{1mm} \forall \hspace{1mm} i,j \in \mathcal{V} \hspace{1mm}\&\hspace{1mm} i\neq j\}$. $\mathcal{A}$ is the weighted adjacency matrix such that $\mathcal{A} = a(i,j) \in \mathbb{R}\textsuperscript{N$\times$N}$.

The possibility of existence of an edge $(i,j)$ occurs iff the vertex $i$ receives the information supplied by the vertex $j$, i.e., $(i,j) \in \mathcal{E}$. Hence, $i$ and $j$ are termed neighbours. The set $\mathcal{N}_i$ contains labels of vertices that are neighbours of the vertex $i$. For the adjacency matrix $\mathcal{A}$, $a(i,j) \in \mathbb{R}^+_0$. If $(i,j) \in \mathcal{E} \Rightarrow a(i,j)>0$. If $(i,j) \notin \mathcal{E}$ or $i = j \Rightarrow a(i,j)=0$.

The Laplacian matrix $\mathcal{L}$ \cite{chung} is central to the consensus problem and is given by $\mathcal{L} = \mathcal{D-A}$ where degree matrix, $\mathcal{D}$ is a diagonal matrix, i.e, $\mathcal{D}$ = diag($d_1,d_2,...,d_n$) whose entries are $d_i = \sum_{j=1}^{n} a(i,j)$. A directed path from vertex $j$ to vertex $i$ defines a sequence comprising of edges $(i,i_1), (i_1,i_2), ..., (i_l,j)$ with distinct vertices $i_k \in \mathcal{V}$, $k = 1, 2, 3, ..., l$. Incidence matrix $\mathcal{B}$ is also a diagonal matrix with entries $1$ or $0$. The entry is $1$ if there exists an edge between leader agent and any other agent, otherwise it is $0$. Furthermore, it can be inferred that the path between two distinct vertices is not uniquely determined. However, if a distinct node in $\mathcal{V}$ contains directed path to every other distinct node in $\mathcal{V}$, then the directed graph $\mathcal{G}$ is said to have a spanning tree. Consequently,the matrix $\mathcal{L} + \mathcal{B}$ has full rank \cite{chung}. Physically, each agent has been modelled by a vertex or node and the line of communication between any two agents has been modelled as a directed edge.

\subsection{Sliding Mode Control}
Sliding Mode Control (SMC) \cite{utkin} is known for its inherent robustness. The switching nature of the control is used to nullify bounded disturbances and matched uncertainties. Switching happens about a hypergeometric manifold in state space known as sliding manifold, surface, or hyperplane. The control drives the system monotonically towards the sliding surface, i.e, trajectories emanate and move towards the hyperplane (reaching phase). System trajectories, after reaching the hyperplane, get constrained there for all future time (sliding phase), thereby ensuring the system dynamics remains independent of bounded disturbances and matched uncertainties.

In order to push state trajectories onto the surface $s(x)$, a proper discontinuous control effort $u_{\textsc{SM}}(t,x)$ needs to be synthesized satisfying the following inequality.
\begin{equation}
s^T(x) \dot{s}(x) \leq -\eta \|s(x)\|,
\end{equation}
with $\eta$ being positive and is referred as the reachability constant.
\begin{equation}
\because \hspace{4mm} \dot{s}(x) = \frac{\partial s}{\partial x} \dot{x} = \frac{\partial s}{\partial x} f(t,x,u_{\textsc{SM}})
\end{equation}
\begin{equation}
\therefore \hspace{4mm} s^T(x) \frac{\partial s}{\partial x} f(t,x,u_{\textsc{SM}}) \leq -\eta \|s(x)\|.
\end{equation}
The motion of state trajectories confined on the manifold is known as \emph{sliding}. Sliding mode exists if the state velocity vectors are directed towards the manifold in its neighbourhood. Under such consideration, the manifold is called attractive, i.e., trajectories starting on it remain there for all future time and trajectories starting outside it tend to it in an asymptotic manner. Hence, in sliding motion,
\begin{equation}
\dot{s}(x) = \frac{\partial s}{\partial x} f(t,x,u_{\textsc{SM}}) = 0.
\end{equation}
$u_{\textsc{SM}} = u_{eq}$ is a solution, generally referred as equivalent control is not the actual control applied to the system but can be thought of as a control that must be applied on an average to maintain sliding motion and is mainly used for analysis of sliding motion. 

\section{Dynamics of Multi-Agent Systems \& Problem Formulation}
Consider first order homogeneous MAS interacting among themselves and their environment in a directed topology. Under such interconnection, information about the predicted location of source of the odor through instantaneous plume sensing is not available globally. However, local information is obtained by communication among agents whenever at least one agent attains some information of interest. The governing dynamics of first order homogeneous MAS consisting of $N$ agents is described by nonlinear differential equations as
\begin{equation}\label{eq:dynamics}
\dot{x}_i(t) = f(x_i(t)) + u_{\textsc{SM}_i}(t) + \varsigma_i;\hspace{2mm} i \in [1,N],
\end{equation}
where $f(\cdot):\mathbb{R}\textsuperscript{+} \times X \rightarrow \mathbb{R}\textsuperscript{m}$ is assumed to be locally Lipschitz over some fairly large domain $\mathbb{D_L}$ with Lipschitz constant $\bar{L}$, and denotes the uncertain nonlinear dynamics of each agent. Also $X \subset \mathbb{R}\textsuperscript{m}$ is a domain in which origin is contained. $x_i$ and $u_{\textsc{SM}_i}$ are the state of $i$\textsuperscript{th} agent and the associated control respectively. $\varsigma_i$ represents bounded exogenous disturbances that enter the system from input channel, i.e., $\|\varsigma_i\| \leq \varsigma_{max} < \infty$.

The problem of odor source localization can be viewed as a cooperative control problem in which control laws $u_{\textsc{SM}_i}$ need to be designed such that the conditions $\lim_{t\to\infty}\| x_i - x_j\| = 0$ and $\lim_{t\to\infty} \| x_i - x_s\|\leq \theta$ are satisfied. Here $x_s$ represents the probable location of odor source \& $\theta$ is an accuracy parameter.

\section{Hierarchical Distributed Cooperative Control Scheme}
In order to drive the agents towards consensus to locate the source of odor, we propose the following hierarchy.
\subsection{Group Decision Making}
This layer utilizes both concentration and wind information to predict the location of odor source. Then, the final probable position of the source can be described as
\begin{equation}\label{eq:xpp}
    \psi(t_k) = c_1 p_i(t_k) + (1 - c_1) q_i(t_k),
\end{equation}
with $p_i(t_k)$ as the oscillation centre according to a simple Particle Swarm Optimization (PSO) algorithm and $q_i(t_k)$ captures the information of the wind. $c_1 \in (0,1)$ denotes additional weighting coefficient.
\begin{remark}\label{sensorDiscreteInstant}
The arguments in {\normalfont (\ref{eq:xpp})} represent data captured at $t=t_k$ instants ($k=1,2,...$) as the sensors equipped with the agents can only receive data at discrete instants.
\end{remark}
It should be noted that $\psi$ is the tracking reference that is fed to the controller. Now, we present detailed description of obtaining $p_i(t_k)$ and $q_i(t_k)$.

Simple PSO algorithm that is commonly used in practice has the following form.
\begin{align}
    v_i(t_{k+1})=\omega v_i(t_k) + u_{\textsc{PSO}}(t_k), \label{vpso}\\
    x_i(t_{k+1})=x_i(t_k) + v_i(t_{k+1}). \label{xpso}
\end{align}
Here $\omega$ is the inertia factor, $v_i(t_k)$ and $x_i(t_k)$ represent the respective velocity and position of $i^{th}$ agent. This commonly used form of PSO can also be used as a proportional-only type controller, however for the disadvantages mentioned earlier, we do not use PSO as our final controller. PSO control law $u_{\textsc{PSO}}$ can be described as
\begin{equation} \label{upso}
    u_{\textsc{PSO}} = \alpha_1 (x_l(t_k) - x_i(t_k)) + \alpha_2 (x_g(t_k) - x_i(t_k)).
\end{equation}
In (\ref{upso}), $x_l(t_k)$ denotes the previous best position and $x_g(t_k)$ denotes the global best position of neighbours of $i^{th}$ agent at time $t=t_k$, and $\alpha_1$ \& $\alpha_2$ are acceleration coefficients. Since, every agent in MAS can get some information about the magnitude of concentration via local communication, position of the agent with a global best can be easily known. By the idea of PSO, we can compute the oscillation centre $p_i(t_k)$ as
\begin{equation}\label{oscillation centre}
    p_i(t_k) = \frac{\alpha_1 x_l(t_k) + \alpha_2 x_g(t_k)}{\alpha_1 + \alpha_2},
\end{equation}
where
\begin{align}
    x_l(t_k) &= \arg \max_{0<t<t_{k-1}} \{ g(x_l(t_{k-1})), g(x_i(t_k)) \}, \\
    x_g(t_k) &= \arg \max_{0<t<t_{k-1}} \{ g(x_g(t_{k-1})), \max_{j\in N} a_{ij}\hspace{1mm} g(x_j(t_k)) \}.
\end{align}
Thus, from (\ref{upso}), (\ref{oscillation centre})
\begin{equation}
    u_{\textsc{PSO}}(t_k) = (\alpha_1 + \alpha_2)\{p_i(t_k)-x_i(t_k)\},
\end{equation}
which is clearly a proportional-only controller with proportional gain $\alpha_1 + \alpha_2$, as highlighted earlier.

In order to compute $q_i(t_k)$, movement process of a single filament that consists several order molecules has been modelled. If $x_f(t)$ denotes position of the filament at time $t$, $\bar{v}_a(t)$ represent mean airflow velocity and $n(t)$ be some random process, then the model can be described as
\begin{equation}\label{eq:filament}
    \dot{x}_f(t) = \bar{v}_a(t) + n(t).
\end{equation}
Without loss of generality, we shall regard the start time of our experiment as $t=0$. From (\ref{eq:filament}), we have
\begin{equation}\label{eq:filamentIntegrated}
    x_f(t) = \int_{0}^{t} \bar{v}_a(\tau) d\tau + \int_{0}^{t} n(\tau) d\tau + x_s(0).
\end{equation}
$x_s(0)$ denotes the real position of the odor source at $t=0$.
\begin{assumption}
We assume the presence of a single, stationary odor source. Thus, $x_s(t) = x_s(0)$.
\end{assumption}
Implications from remark \ref{sensorDiscreteInstant} require (\ref{eq:filamentIntegrated}) to be implemented at $t=t_k$ instants. Hence,
\begin{align}
    x_f(t_k) &= \sum_{m=0}^{t} \bar{v}_a(\tau_m) \Delta t + \sum_{m=0}^{t} n(\tau_m) \Delta t + x_s(t_k), \\
    x_f(t_k) &= x_s(t_k) + \bar{v}_a^\star(t_k) + w^\star (t_k).\label{eq:xf}
\end{align}
In (\ref{eq:xf}), $\sum_{m=0}^{t} \bar{v}_a(\tau_m) \Delta t=\bar{v}_a^\star(t_k)$ and $\sum_{m=0}^{t} n(\tau_m) \Delta t=w^\star (t_k)$.
\begin{remark}
In (\ref{eq:xf}), the accumulated average of $\bar{v}_a^\star(t_k)$ and $w^\star (t_k)$ can also be considered $\forall$ possible filament releasing time.
\end{remark}
From (\ref{eq:xf}),
\begin{equation}
 x_f(t_k) - \bar{v}_a^\star(t_k) = x_s(t_k) + w^\star (t_k).\label{eq:xf2}
\end{equation}
The above relationship, (\ref{eq:xf2}) can be viewed as the information about $x_s(t_k)$ with some noise $w^\star (t_k)$. Hence,
\begin{equation}
q_i(t_k) = x_s(t_k) + w^\star (t_k).\label{eq:qi}
\end{equation}
Therefore, $\psi$ in (\ref{eq:xpp}) can now be constructed from (\ref{oscillation centre}) \& (\ref{eq:qi}).

\subsection{Path Planning}
Since, detection of information of interest is tied to the threshold value defined for the sensors, the next state is updated taking this threshold value into account. Thus, the blueprints of path planning can be described in terms of three types of behavior.
\begin{enumerate}
\item Surging: If the $i^{th}$ agent receives data well above threshold, we say that some clues about the location of the source has been detected. If the predicted position of the source at $t=t_k$ as seen by $i^{th}$ agent be given as $x_{s_i}(t_k)$, then the next state of the agent is given mathematically as
\begin{equation} \label{surging}
x_i(t_{k+1}) = x_{s_i}(t_k).
\end{equation}
\item Casting: If the $i^{th}$ agent fails to detect information at any particular instant, then the next state is obtained using the following relation.
\begin{equation}\label{casting}
x_i(t_{k+1}) = \frac{\| x_i(t_k) - x_{s_i}(t_k)\|}{2} + x_{s_i}(t_k).
\end{equation}
\item Search and exploration: If all the agents fail to detect odor clues for a time segment $[t_k, t_{k+l}]>\delta_0$ for some $l \in \mathbb{N}$ and $\delta_0 \in \mathbb{R}^+$ being the time interval for which no clues are detected or some constraint on wait time placed at the start of the experiment, then the next state is updated as
\begin{equation}\label{searching}
x_i(t_{k+1}) = x_{s_i}(t_k) + \digamma_\sigma^\phi.
\end{equation}
In (\ref{searching}), $\digamma_\sigma^\phi$ is some random parameter with $\sigma$ as its standard deviation and $\phi$ as its mean.
\end{enumerate}

\subsection{Distributed Control}
In the control layer, we design a robust and powerful controller on the paradigms of sliding mode. It is worthy to mention that based on instantaneous sensing and swarm information, at different times, each agent can take up the role of a virtual leader whose opinion needs to be kept by other agents. $\psi$ from (\ref{eq:xpp}) has been provided to the controller as the reference to be tracked. The tracking error is formulated as
\begin{equation}\label{eq:error}
e_i(t) = x_i(t) - \psi(t_k) \hspace{1mm}; \hspace{2mm} t \in [t_k, t_{k+1}[.
\end{equation}
In terms of graph theory, we can reformulate the error variable as
\begin{equation}\label{eq:topological error}
\epsilon_i(t) = (\mathcal{L}+\mathcal{B})e_i(t) =(\mathcal{L}+\mathcal{B}) (x_i(t) - \psi(t_k)).
\end{equation}
From this point onward, we shall denote $\mathcal{L}+\mathcal{B}$ as $\mathcal{H}$. Next, we formulate the sliding manifold
\begin{equation}\label{surface}
s_i(t) = \lambda_1 \tanh(\lambda_2 \epsilon_i(t)),
\end{equation}
which is a nonlinear sliding manifold offering faster reachability to the surface. $\lambda_1 \in \mathbb{R}^+$ represents the speed of convergence to the surface, and $\lambda_2 \in \mathbb{R}^+$ denotes the slope of the nonlinear sliding manifold. These are coefficient weighting parameters that affect the system performance. The forcing function has been taken as
\begin{equation}\label{sdot}
\dot{s}_i(t) = -\mu \sinh^{-1}(m + w|s_i(t)|) sign(s_i(t)).
\end{equation}
In (\ref{sdot}), $m$ is a small offset such that the argument of $\sinh^{-1}$ function remains non zero and $w$ is the gain of the controller. The parameter $\mu$ facilitates additional gain tuning. In general, $m<<w$. This novel reaching law contains a nonlinear gain and provides faster convergence towards the manifold. Moreover, this reaching law is smooth and chattering free, which is highly desirable in mechatronic systems to ensure safe operation.
\begin{theorem}\label{th:usm}
Given the dynamics of MAS {\normalfont(\ref{eq:dynamics})} connected in a directed topology, error candidates {\normalfont(\ref{eq:error}, \ref{eq:topological error})} and the sliding manifold {\normalfont(\ref{surface})}, the stabilizing control law that ensures accurate reference tracking under consensus can be described as
\begin{align}\label{eq:usm}
    u_{\textsc{SM}_i}(t) &= -\big\{(\Lambda \mathcal{H})^{-1} \mu \sinh^{-1}(m + w|s_i(t)|)sign(s_i(t))\Gamma^{-1} \nonumber \\
    &+ (f(x_i(t)) - \dot{\psi}(t_k)) \big\}
\end{align}
where $\Lambda = \lambda_1 \lambda_2$, $\Gamma=1-\tanh^2(\lambda_2\epsilon_i(t))$, $w > \sup_{t\geq0} \{\|\varsigma_i\|\}$ \& $\mu > \sup \{ \|\Lambda \mathcal{H} \varsigma_i \Gamma\| \}$.
\begin{remark}
As mentioned earlier, $\lambda_1, \lambda_2 \in \mathbb{R}^+$. This ensures $\Lambda \neq 0$ and hence its non singularity. The argument of $\tanh$ is always finite and satisfies $\lambda_2\epsilon_i(t)\neq \pi \iota (\kappa + 1/2)$ for $\kappa \in \mathbb{Z}$, thus $\Gamma$ is also invertible. Moreover the non singularity of $\mathcal{H}$ can be established directly if the digraph contains a spanning tree with leader agent as a root.
\end{remark}
\end{theorem}
\begin{proof}
From (\ref{eq:topological error}) and (\ref{surface}), we can write
\begin{align}
    \dot{s}_i(t) &= \lambda_1\{ \lambda_2 \dot{\epsilon}_i(t) (1 - \tanh^2(\lambda_2\epsilon_i(t)))\} \\
    &=\lambda_1 \lambda_2 \dot{\epsilon}_i(t) - \lambda_1\lambda_2 \dot{\epsilon}_i(t)\tanh^2(\lambda_2\epsilon_i(t))\\
    &=\lambda_1 \lambda_2 \dot{\epsilon}_i(t)\{1-\tanh^2(\lambda_2\epsilon_i(t))\}\\
    &=\Lambda \mathcal{H}(\dot{x}_i(t) - \dot{\psi}(t_k)) \Gamma \label{sdotIn}
\end{align}  
with $\Lambda$ \& $\Gamma$ as defined in Theorem \ref{th:usm}. From (\ref{eq:dynamics}), (\ref{sdotIn}) can be further simplified as
\begin{align}
    \dot{s}_i(t) &=\Lambda \mathcal{H}(f(x_i(t)) + u_{\textsc{SM}_i}(t) + \varsigma_i - \dot{\psi}(t_k)) \Gamma.
\end{align}
Using (\ref{sdot}), the control that brings the state trajectories on to the sliding manifold can now be written as
\begin{align}
    u_{\textsc{SM}_i}(t) &= -\big\{(\Lambda \mathcal{H})^{-1} \mu \sinh^{-1}(m + w|s_i(t)|)sign(s_i(t))\Gamma^{-1} \nonumber \\
    &+ (f(x_i(t)) - \dot{\psi}(t_k))\big\}.
\end{align}
This concludes the proof.
\end{proof}
\begin{remark}
The control {\normalfont (\ref{eq:usm})} can be practically implemented as it does not contain the uncertainty term.
\end{remark}
It is crucial to analyze the necessary and sufficient conditions for the existence of sliding mode when control protocol (\ref{eq:usm}) is used. We regard the system to be in sliding mode if for any time $t_1 \in [0, \infty[$, system trajectories are brought upon the manifold $s_i(t)=0$ and are constrained there for all time thereafter, i.e., for $t\geq t_1$, sliding motion occurs.
\begin{theorem}
Consider the system described by {\normalfont(\ref{eq:dynamics})}, error candidates {\normalfont(\ref{eq:error}, \ref{eq:topological error})}, sliding manifold {\normalfont(\ref{surface})} and the control protocol {\normalfont(\ref{eq:usm})}. Sliding mode is said to exist in vicinity of sliding manifold, if the manifold is attractive, i.e., trajectories emanating outside it continuously decrease towards it. Stating alternatively, reachability to the surface is ensured for some reachability constant $\eta>0$. Moreover, stability can be guaranteed in the sense of Lyapunov if gain $\mu$ is designed as $\mu > \sup \{ \|\Lambda \mathcal{H} \varsigma_i \Gamma\| \}$.
\end{theorem}
\begin{proof}
Let us take into account, a Lyapunov function candidate
\begin{equation}\label{lyapunov}
    V_i = 0.5s^2_i.
\end{equation}
Taking derivative of (\ref{lyapunov}) along system trajectories yield
\begin{align}
    \dot{V}_i &= s_i \dot{s}_i\\
    &=s_i \big \{\Lambda \mathcal{H}(f(x_i(t)) + u_{\textsc{SM}_i}(t) + \varsigma_i - \dot{\psi}(t_k)) \Gamma \big\} \label{ssdot}.
\end{align}
Substituting the control protocol (\ref{eq:usm}) in (\ref{ssdot}), we have
\begin{align}
\dot{V}_i &=s_i \big(-\mu \sinh^{-1}(m+w|s_i|)sign(s_i) + \Lambda \mathcal{H} \varsigma_i \Gamma \big) \nonumber \\
&= -\mu \sinh^{-1}(m+w|s_i|)\|s_i\| + \Lambda \mathcal{H} \varsigma_i \Gamma \|s_i\| \nonumber \\
&=\big \{-\mu \sinh^{-1}(m+w|s_i|) + \Lambda \mathcal{H} \varsigma_i \Gamma \big \} \|s_i\| \nonumber \\
&=-\eta \|s_i\|,
\label{ssdotfinal}
\end{align}
where $\eta = \mu \sinh^{-1}(m+w|s_i|) - \Lambda \mathcal{H} \varsigma_i \Gamma > 0$ is called reachability constant. For $\mu > \sup \{ \|\Lambda \mathcal{H} \varsigma_i \Gamma\| \}$, we have
\begin{equation}
    \dot{V}_i < 0.
\end{equation}
Thus, the derivative of Lyapunov function candidate is negative definite confirming stability in the sense of Lyapunov.\\
Since, $\mu>0$, $\|s_i\|>0$ and $\sinh^{-1}(\cdot)>0$ due to the nature of its arguments. Therefore, (\ref{ssdotfinal}) and (\ref{sdot}) together provide implications that $\forall s_i(0)$, $s_i \dot{s}_i<0$ and the surface is globally attractive. This ends the proof.
\end{proof}

\section{Results and discussions}
Interaction topology of the agents represented as a digraph has been shown here in figure 1. The associated graph matrices have been described below. The computer simulation has been performed assuming that agent 1 appears as virtual leader to all other agents, making the topology fixed and directed for this study. It should be noted that, the theory developed so far can be extended to the case of switching topologies and shall be dealt in future.
\begin{figure}[h!]
\centering
  \includegraphics[scale=0.8]{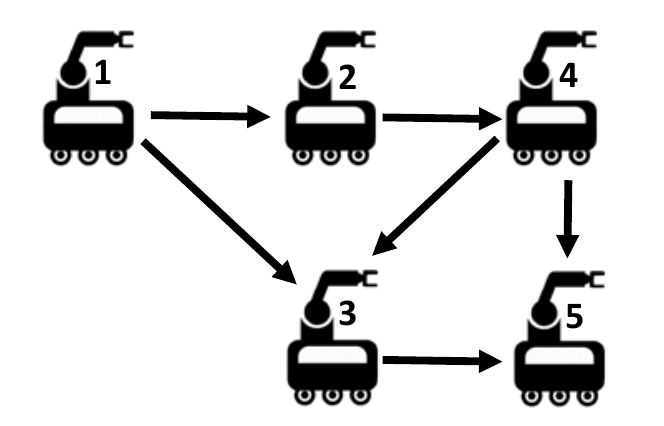}
  \caption{Topology in which agents are connected}
\end{figure}
\begin{equation}
\mathcal{A}=
\begin{bmatrix}
0 & 0 & 1 & 0\\
0 & 0 & 0 & 0\\
0 & 1 & 0 & 0\\
0 & 0 & 1 & 0
\end{bmatrix}
,\hspace{1.5mm}
\mathcal{B}=
\begin{bmatrix}
1 & 0 & 0 & 0\\
0 & 1 & 0 & 0\\
0 & 0 & 0 & 0\\
0 & 0 & 0 & 0
\end{bmatrix}
,\hspace{1.5mm}
\mathcal{D}=
\begin{bmatrix}
1 & 0 & 0 & 0\\
0 & 0 & 0 & 0\\
0 & 0 & 1 & 0\\
0 & 0 & 0 & 1
\end{bmatrix},
\end{equation}
\begin{equation}
\mathcal{L}=\mathcal{D}-\mathcal{A}=
\begin{bmatrix}
1 & 0 & -1 & 0\\
0 & 0 &  0 & 0\\
0 & -1 & 1 & 0\\
0 & 0 & -1 & 1
\end{bmatrix},
\mathcal{L+B}=
\begin{bmatrix}
2 & 0 & -1 & 0\\
0 & 1 & 0 & 0\\
0 & -1 & 1 & 0\\
0 & 0 & -1 & 1
\end{bmatrix}
\end{equation}
Agents have the following dynamics.
\begin{align}
\dot{x}_1 &= 0.1 \sin(x_1) + \cos (2 \pi t) + u_{\textsc{SM}_1}(t) + \varsigma_1,\\
\dot{x}_2 &= 0.1 \sin(x_2) + \cos (2 \pi t) + u_{\textsc{SM}_2}(t) + \varsigma_2,\\
\dot{x}_3 &= 0.1 \sin(x_3) + \cos (2 \pi t) + u_{\textsc{SM}_3}(t) + \varsigma_3,\\
\dot{x}_4 &= 0.1 \sin(x_4) + \cos (2 \pi t) + u_{\textsc{SM}_4} (t)+ \varsigma_4, \\
\dot{x}_5 &= 0.1 \sin(x_5) + \cos (2 \pi t) + u_{\textsc{SM}_5} (t)+ \varsigma_5.
\end{align}
In this study, advection model given in \cite{6640740} has been used to simulate the plume with both additive and multiplicative disturbances.
The initial conditions for simulation are taken to be large values, i.e., far away from the equilibrium point. Time varying disturbance has been taken as $\varsigma_i = 0.3 \sin (\pi ^2 t^2)$, accuracy parameter $\theta = 0.001$ and maximum mean airflow velocity $\bar{v}_{a_{max}} = 1$ m/s. Other key design parameters are mentioned in table 1.
\begin{figure*}[h!]
  \includegraphics[width=\textwidth]{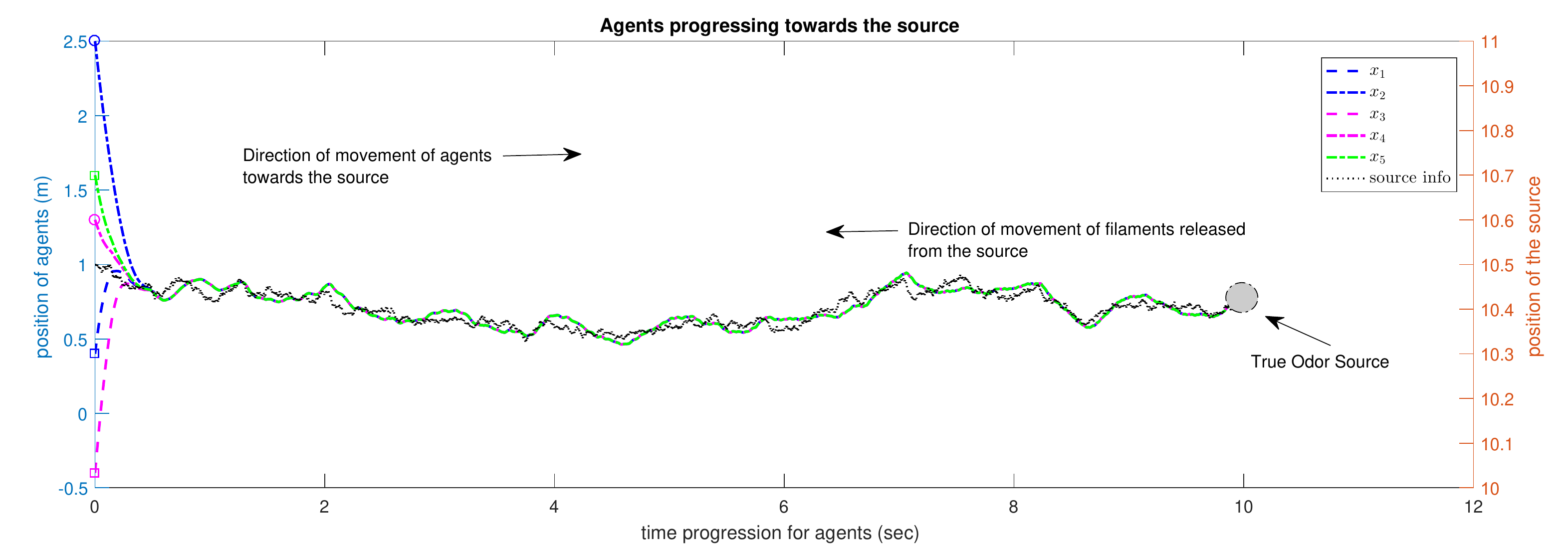}
  \caption{Agents in consensus to locate source of odor}
 \includegraphics[width=\textwidth]{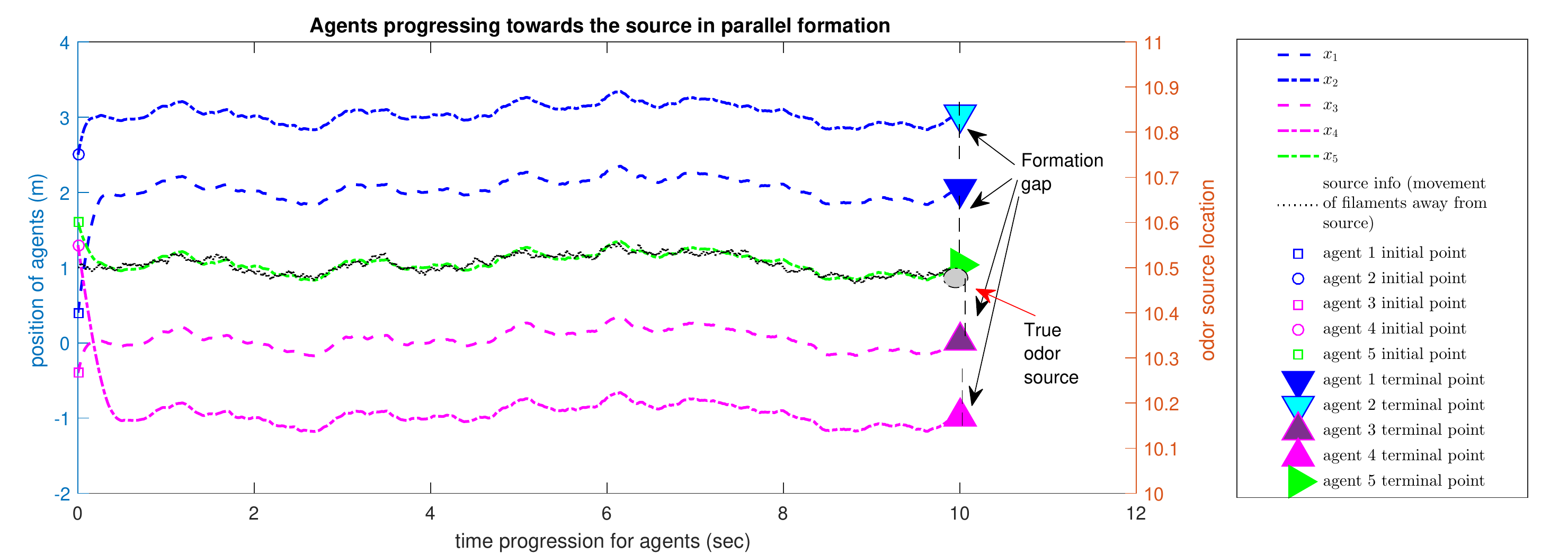}
  \caption{Agents in formation to locate source of odor}
\end{figure*}
\begin{figure}[H]
\includegraphics[width=0.5\textwidth]{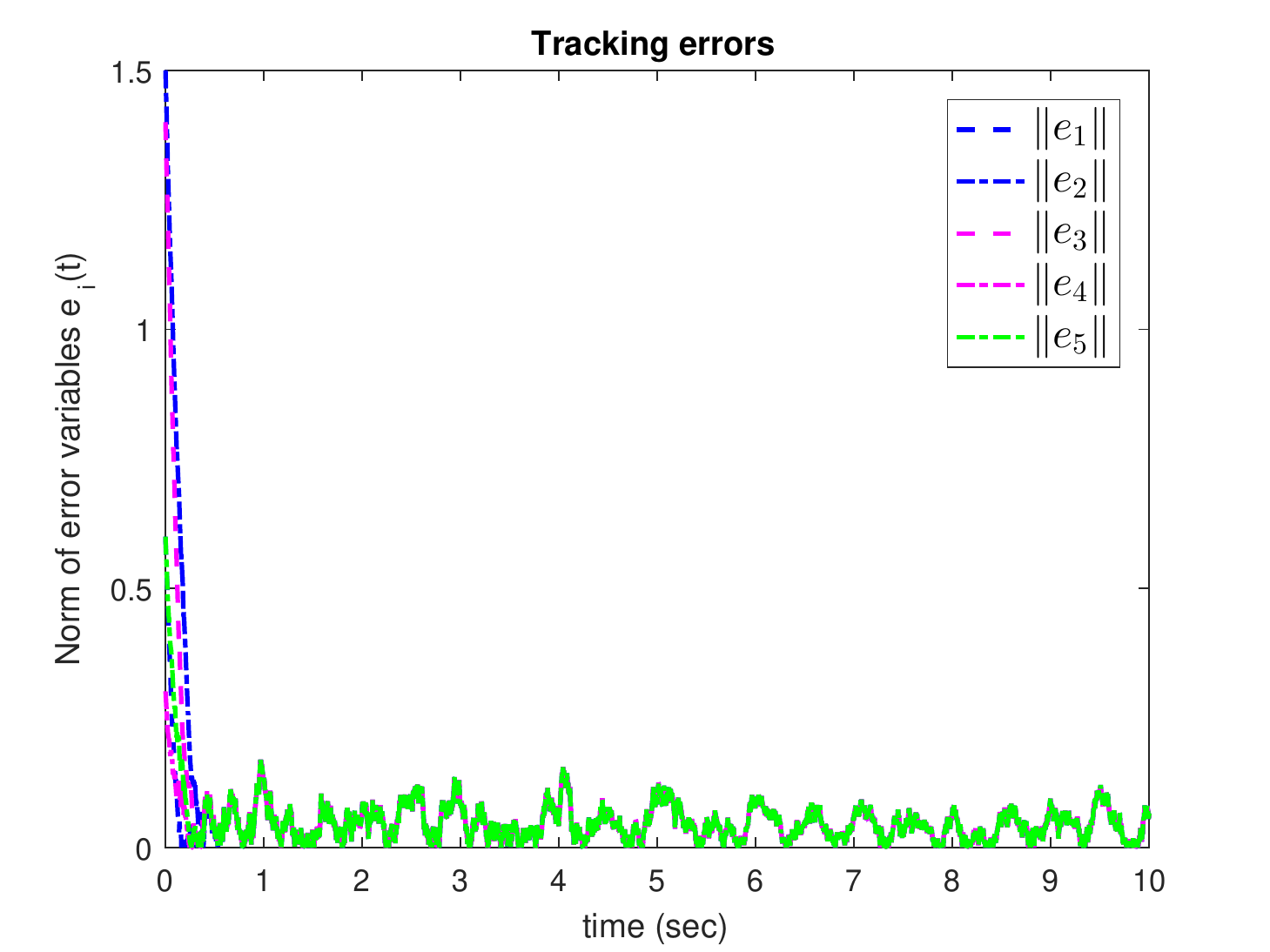}
  \caption{Norm of tracking errors}
\end{figure}
\begin{figure}[H]
  \includegraphics[width=0.5\textwidth]{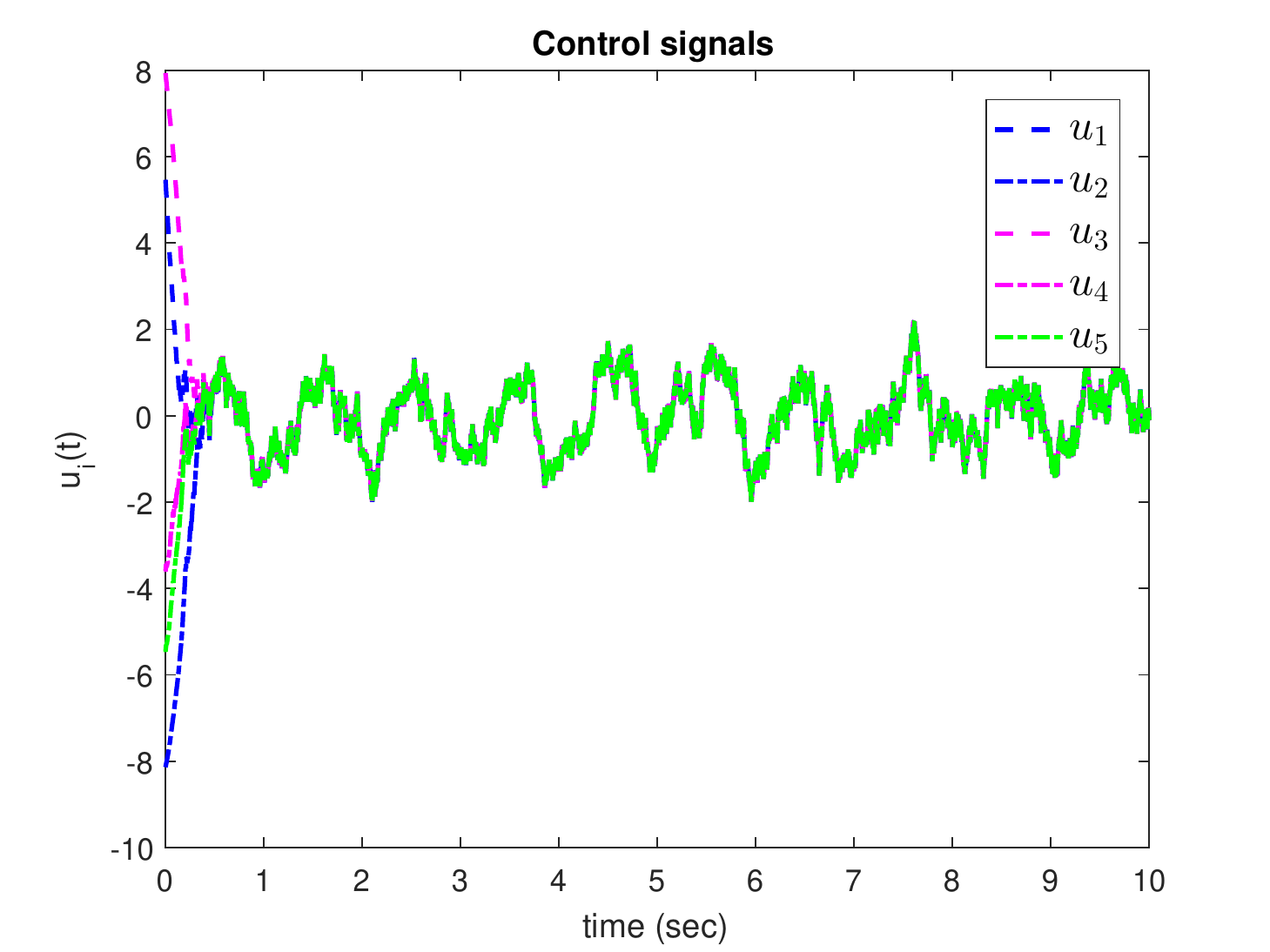}
  \caption{Control signals during consensus}
\end{figure}
\begin{figure}[H]
\includegraphics[width=0.5\textwidth]{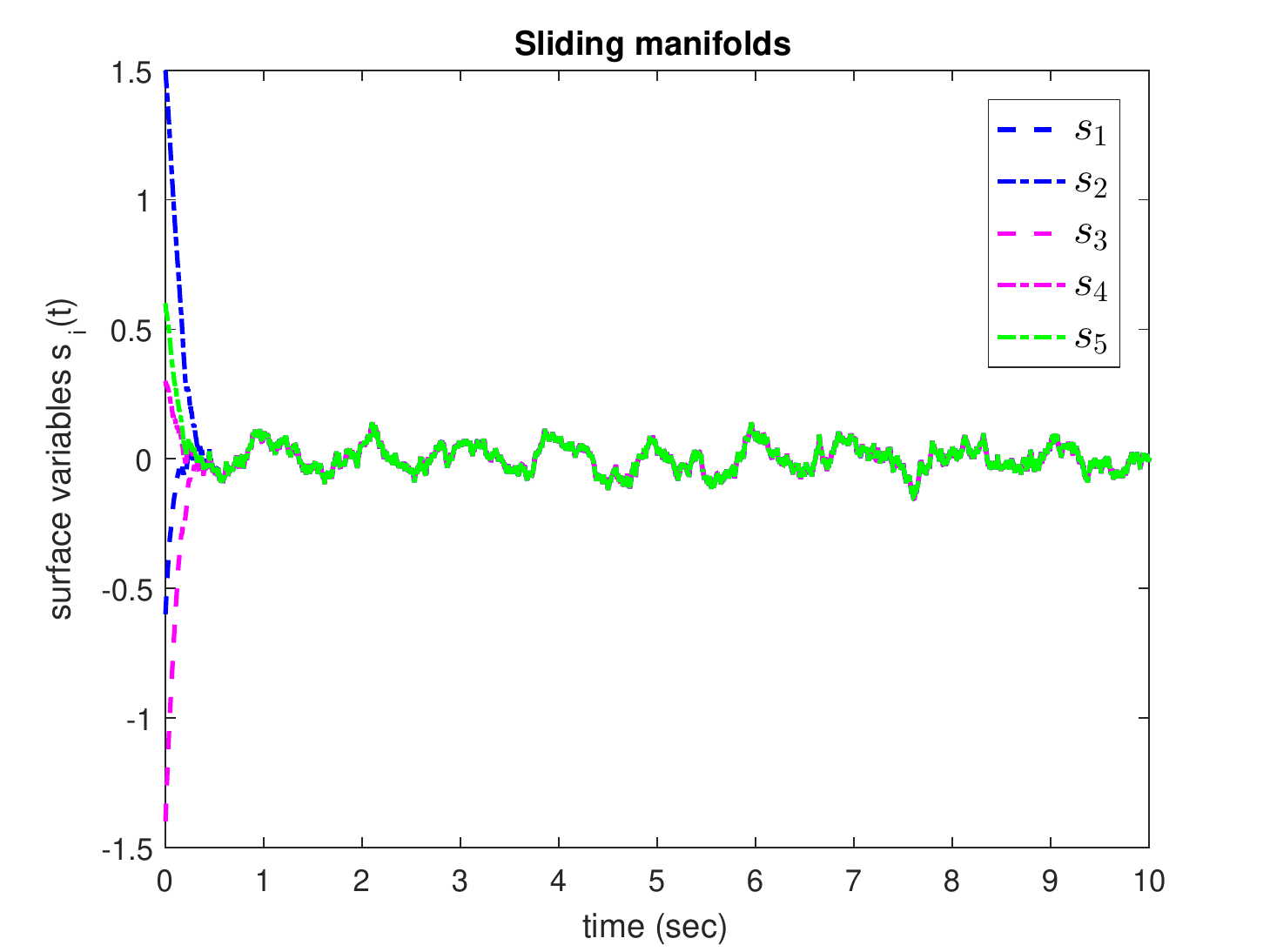}
  \caption{Sliding manifolds during consensus}
\end{figure}
\begin{table}[H]
\caption{Values of the design parameters used in simulation}\label{tb:parameter}
\centering
\begin{tabular}{ ccccccccc } 
\hline
 $c_1$ & $\omega_{max}$ & $\alpha_1$ & $\alpha_2$ & $\lambda_1$ & $\lambda_2$ & $\mu$ & $m$ & $w$\\\hline\\
 $0.5$ & $2$ rad/s & $0.25$ & $0.25$ & $1.774$ & $2.85$ & $5$& $10^{-3}$ & $2$
\end{tabular}
\end{table}
Figure 2 shows agents coming to consensus in finite time to locate the source of odor and figure 3 shows agents moving in parallel formation to locate the odor source. Norm of the tracking errors has been depicted in figure 4. It is evident that the magnitude of error is very small. Plot of control signals during consensus has been shown in figure 5 and the plot of sliding manifolds has been shown in figure 6.

\section{Concluding remarks}
The problem of odor source localization by MAS has been dealt with in a hierarchical manner in this work. The problem translates into a cooperative control problem wherein agents are driven towards consensus to locate the true odor source in finite time. Through computer simulations, it has been confirmed that the proposed strategy is faster and provides accurate tracking even in the presence of time varying disturbances.

\bibliographystyle{IEEEtran}
\bibliography{references}

\begin{thebibliography}{10}
\providecommand{\url}[1]{#1}
\csname url@samestyle\endcsname
\providecommand{\newblock}{\relax}
\providecommand{\bibinfo}[2]{#2}
\providecommand{\BIBentrySTDinterwordspacing}{\spaceskip=0pt\relax}
\providecommand{\BIBentryALTinterwordstretchfactor}{4}
\providecommand{\BIBentryALTinterwordspacing}{\spaceskip=\fontdimen2\font plus
\BIBentryALTinterwordstretchfactor\fontdimen3\font minus
  \fontdimen4\font\relax}
\providecommand{\BIBforeignlanguage}[2]{{%
\expandafter\ifx\csname l@#1\endcsname\relax
\typeout{** WARNING: IEEEtran.bst: No hyphenation pattern has been}%
\typeout{** loaded for the language `#1'. Using the pattern for}%
\typeout{** the default language instead.}%
\else
\language=\csname l@#1\endcsname
\fi
#2}}
\providecommand{\BIBdecl}{\relax}
\BIBdecl

\bibitem{Larcombe1984}
M.~H.~E. Larcombe, \emph{Robotics in nuclear engineering: Computer assisted
  teleoperation in hazardous environments with particular reference to
  radiation fields}.\hskip 1em plus 0.5em minus 0.4em\relax United States:
  Graham and Trotman, Inc, 1984.

\bibitem{240354}
R.~Rozas, J.~Morales, and D.~Vega, ``Artificial smell detection for robotic
  navigation,'' in \emph{Advanced Robotics, 1991. 'Robots in Unstructured
  Environments', 91 ICAR., Fifth International Conference on}, June 1991, pp.
  1730--1733 vol.2.

\bibitem{594225}
V.~Genovese, P.~Dario, R.~Magni, and L.~Odetti, ``Self organizing behavior and
  swarm intelligence in a pack of mobile miniature robots in search of
  pollutants,'' in \emph{Proceedings of the IEEE/RSJ International Conference
  on Intelligent Robots and Systems}, vol.~3, Jul 1992, pp. 1575--1582.

\bibitem{Buscemi}
L.~Buscemi, M.~Prati, and G.~Sandini, ``Cellular robotics: Behaviour in
  polluted environments,'' in \emph{Proceedings of the 2nd International
  Symposium on Distributed Autonomous Robotic Systems}, 1994.

\bibitem{414920}
R.~A. Russell, ``Laying and sensing odor markings as a strategy for assisting
  mobile robot navigation tasks,'' \emph{IEEE Robotics Automation Magazine},
  vol.~2, no.~3, pp. 3--9, Sep 1995.

\bibitem{Russell:2000:ODM:518667}
{Russell, R. Andrew}, \emph{Odour Detection by Mobile Robots}.\hskip 1em plus
  0.5em minus 0.4em\relax River Edge, NJ, USA: World Scientific Publishing Co.,
  Inc., 2000.

\bibitem{RUSSELL200383}
\BIBentryALTinterwordspacing
R.~Russell, A.~Bab-Hadiashar, R.~L. Shepherd, and G.~G. Wallace, ``A comparison
  of reactive robot chemotaxis algorithms,'' \emph{Robotics and Autonomous
  Systems}, vol.~45, no.~2, pp. 83 -- 97, 2003. [Online]. Available:
  \url{http://www.sciencedirect.com/science/article/pii/S0921889003001209}
\BIBentrySTDinterwordspacing

\bibitem{vergassola:hal-00326807}
\BIBentryALTinterwordspacing
M.~Vergassola, E.~Villermaux, and B.~Shraiman, ``{''Infotaxis'' as a strategy
  for searching without gradients},'' \emph{{Nature}}, vol. 445, no. 7126, pp.
  406--409, 2007. [Online]. Available:
  \url{https://hal.archives-ouvertes.fr/hal-00326807}
\BIBentrySTDinterwordspacing

\bibitem{1245262}
J.~A. Farrell, S.~Pang, and W.~Li, ``Plume mapping via hidden markov methods,''
  \emph{IEEE Transactions on Systems, Man, and Cybernetics, Part B
  (Cybernetics)}, vol.~33, no.~6, pp. 850--863, Dec 2003.

\bibitem{1703649}
S.~Pang and J.~A. Farrell, ``Chemical plume source localization,'' \emph{IEEE
  Transactions on Systems, Man, and Cybernetics, Part B (Cybernetics)},
  vol.~36, no.~5, pp. 1068--1080, Oct 2006.

\bibitem{Zarzhitsky:2008:PAC:1571283}
D.~Zarzhitsky, ``Physics-based approach to chemical source localization using
  mobile robotic swarms,'' Ph.D. dissertation, 2008.

\bibitem{braitenberg}
V.~Braitenberg, \emph{Vehicles: Experiments in Synthetic Psychology}.\hskip 1em
  plus 0.5em minus 0.4em\relax Boston, MA, USA: MIT Press, 1984.

\bibitem{Lytridis:2001:ONS:770673.770677}
\BIBentryALTinterwordspacing
C.~Lytridis, G.~S. Virk, Y.~Rebour, and E.~E. Kadar, ``Odor-based navigational
  strategies for mobile agents,'' \emph{Adapt. Behav.}, vol.~9, no. 3-4, pp.
  171--187, Apr. 2001. [Online]. Available:
  \url{http://dx.doi.org/10.1177/10597123010093004}
\BIBentrySTDinterwordspacing

\bibitem{ISHIDA1994153}
\BIBentryALTinterwordspacing
H.~Ishida, K.~Suetsugu, T.~Nakamoto, and T.~Moriizumi, ``Study of autonomous
  mobile sensing system for localization of odor source using gas sensors and
  anemometric sensors,'' \emph{Sensors and Actuators A: Physical}, vol.~45,
  no.~2, pp. 153 -- 157, 1994. [Online]. Available:
  \url{http://www.sciencedirect.com/science/article/pii/0924424794008299}
\BIBentrySTDinterwordspacing

\bibitem{russell_silkworm}
R.~Russell, \emph{Chemical source location and the RoboMole project}.\hskip 1em
  plus 0.5em minus 0.4em\relax Australian Robotics and Automation Association,
  2003, pp. 1 -- 6.

\bibitem{862824}
L.~Marques and A.~T.~D. Almeida, ``Electronic nose-based odour source
  localization,'' in \emph{6th International Workshop on Advanced Motion
  Control. Proceedings (Cat. No.00TH8494)}, April 2000, pp. 36--40.

\bibitem{MARQUES200251}
\BIBentryALTinterwordspacing
L.~Marques, U.~Nunes, and A.~T. de~Almeida, ``Olfaction-based mobile robot
  navigation,'' \emph{Thin Solid Films}, vol. 418, no.~1, pp. 51 -- 58, 2002,
  proceedings from the International School on Gas Sensors in conjunction with
  the 3rd European School of the NOSE Network. [Online]. Available:
  \url{http://www.sciencedirect.com/science/article/pii/S004060900200593X}
\BIBentrySTDinterwordspacing

\bibitem{1431045}
W.~Ren and R.~W. Beard, ``Consensus seeking in multiagent systems under
  dynamically changing interaction topologies,'' \emph{IEEE Transactions on
  Automatic Control}, vol.~50, no.~5, pp. 655--661, May 2005.

\bibitem{5711689}
W.~Yu, G.~Chen, W.~Ren, J.~Kurths, and W.~X. Zheng, ``Distributed higher order
  consensus protocols in multiagent dynamical systems,'' \emph{IEEE
  Transactions on Circuits and Systems I: Regular Papers}, vol.~58, no.~8, pp.
  1924--1932, Aug 2011.

\bibitem{Hayes:2003:SRO:976910.976917}
\BIBentryALTinterwordspacing
A.~T. Hayes, A.~Martinoli, and R.~M. Goodman, ``Swarm robotic odor
  localization: Off-line optimization and validation with real robots,''
  \emph{Robotica}, vol.~21, no.~4, pp. 427--441, Aug. 2003. [Online].
  Available: \url{http://dx.doi.org/10.1017/S0263574703004946}
\BIBentrySTDinterwordspacing

\bibitem{488968}
J.~Kennedy and R.~Eberhart, ``Particle swarm optimization,'' in \emph{Neural
  Networks, 1995. Proceedings., IEEE International Conference on}, vol.~4, Nov
  1995, pp. 1942--1948 vol.4.

\bibitem{Marques2006}
\BIBentryALTinterwordspacing
L.~Marques, U.~Nunes, and A.~T. de~Almeida, ``Particle swarm-based olfactory
  guided search,'' \emph{Autonomous Robots}, vol.~20, no.~3, pp. 277--287, Jun
  2006. [Online]. Available: \url{https://doi.org/10.1007/s10514-006-7567-0}
\BIBentrySTDinterwordspacing

\bibitem{4168420}
W.~Jatmiko, K.~Sekiyama, and T.~Fukuda, ``A pso-based mobile robot for odor
  source localization in dynamic advection-diffusion with obstacles
  environment: theory, simulation and measurement,'' \emph{IEEE Computational
  Intelligence Magazine}, vol.~2, no.~2, pp. 37--51, May 2007.

\bibitem{5586502}
Q.~Lu, S.~r.~Liu, and X.~n.~Qiu, ``A distributed architecture with two layers
  for odor source localization in multi-robot systems,'' in \emph{IEEE Congress
  on Evolutionary Computation}, July 2010, pp. 1--7.

\bibitem{6119291}
Q.~Lu and Q.~L. Han, ``Decision-making in a multi-robot system for odor source
  localization,'' in \emph{IECON 2011 - 37th Annual Conference of the IEEE
  Industrial Electronics Society}, Nov 2011, pp. 74--79.

\bibitem{chung}
\BIBentryALTinterwordspacing
{Fan R. K. Chung}, \emph{Spectral Graph Theory}, ser. CBMS Regional Conference
  Series in Mathematics.\hskip 1em plus 0.5em minus 0.4em\relax AMS and CBMS,
  1997, vol.~92. [Online]. Available: \url{http://bookstore.ams.org/cbms-92}
\BIBentrySTDinterwordspacing

\bibitem{utkin}
{K. David Young, Vadim I. Utkin and Umit Ozguner}, ``A control engineer's guide
  to sliding mode control,'' \emph{IEEE transactions on Control Systems
  Technology}, vol.~7, no.~3, pp. 328--342, May 1999.

\bibitem{6640740}
M.~L. Cao, Q.~H. Meng, Y.~X. Wu, M.~Zeng, and W.~Li, ``Consensus based
  distributed concentration-weighted summation algorithm for gas-leakage source
  localization using a wireless sensor network,'' in \emph{Proceedings of the
  32nd Chinese Control Conference}, July 2013, pp. 7398--7403.

\end{thebibliography}
\end{document}